\documentclass[prx,twocolumn]{revtex4-2} 

\usepackage{yfonts}
\usepackage{amsfonts,amscd,mathrsfs,amsmath,amsthm,amssymb}
\usepackage{mathtools}
\usepackage{tocloft}
\usepackage{xparse}
\usepackage{enumerate}
\usepackage[shortlabels]{enumitem}
\usepackage{graphicx}
\usepackage{float}
\UseRawInputEncoding

\PassOptionsToPackage{hyphens}{url}\usepackage{hyperref}
\usepackage[capitalise]{cleveref}

\usepackage{physics}
\usepackage{bm}
\usepackage{bbm}
\usepackage{caption}

\hypersetup{
    colorlinks=true,
    linkcolor=red,
    filecolor=black,      
    citecolor=red,
}

\newtheorem{theorem}{Theorem}

\newtheorem{corollary}{Corollary}
\newtheorem{lemma}{Lemma}

\newtheorem{prop}{Proposition}
 
    \newenvironment{pf}[1][\textit{Pf}]{ \proof[#1]\ }{\endproof}

    \newcommand\numberthis{\addtocounter{equation}{1}\tag{\theequation}}

    \allowdisplaybreaks

    \newcommand{\bfit}[1]{\textcolor{black}{\textit{\textbf{#1}}}}

    \NewDocumentCommand\explainequals{m}{\overset{\textit{#1}}{=}}

    \let\oldcheck\checkmark
    \renewcommand{\checkmark}{\textcolor{green}{\text{\oldcheck}} }

    \newcommand{\hrulechapter}{\ \hrule width \hsize \kern 1mm \hrule width \hsize height 0.8pt}
    
\begin{document}

\title{Quantum Weight Enumerators for Real Codes \\  with $X$ and $Z$ Exactly Transversal}

\author{Eric Kubischta}
\email{erickub@umd.edu}
\author{Ian Teixeira}
\email{igt@umd.edu}
\author{J.\ Maxwell Silvester}
\email{jsilves1@uci.edu}

\begin{abstract}
In this note we show that the weight enumerators of a real quantum error correcting code with $ X $ and $ Z $ exactly transversal must satisfy certain identities. One consequence of these identities is that if the code is error detecting then it is automatically error correcting for free; implying a relationship between transversality and code distance.
\end{abstract}

\maketitle

\newpage 
\section{Introduction}

Let $((n,K,d))$ denote an $n$-qubit quantum error correcting code with a codespace of dimension $K$ and with a distance of $d$. 

In this paper we will study quantum weight enumerators \cite{quantumMacWilliams} for $K = 2$ codes under the assumption that the codes are \textit{real} and the Pauli gates $ X = \smqty(0 & 1 \\ 1 & 0)$ and $ Z = \smqty(1 & 0 \\ 0 & -1)$ are \textit{exactly} transversal. We say a code is \bfit{real} if the orthogonal projector onto the codespace is real, $ \Pi^*=\Pi $. Equivalently a code is real if there exists a basis of codewords that has all real coefficients with respect to the computational basis.

For a given code we say that a gate $U$ is \bfit{exactly transversal} if $U^{\otimes n}$ on the physical space implements logical $U$ on the code space. 

These assumptions seem restrictive but in fact many familiar codes are real and have $ X $ and $ Z $ exactly transversal. Some examples include: the $[[5,1,3]]$ code, the $[[7,1,3]]$ Steane code, Shor's $[[9,1,3]]$ code, the $[[11,1,5]]$ code, the $ [[15,1,3]] $ code and indeed the whole family of $ [[2^r-1,1,3]] $ quantum Reed-Muller codes.

\section{Preliminaries}

\subsection{Weight Enumerators}

Let $\Pi$ be the code projector. The weight enumerator and the dual weight enumerator are defined in \cite{quantumMacWilliams} as
\begin{align*}
    A_i &= \frac{1}{K^2} \sum_{E \in \mathcal{E}_i} \Tr{ E \Pi}^2, \\
    B_i &= \frac{1}{K} \sum_{E \in \mathcal{E}_i} \Tr{E \Pi E \Pi},
\end{align*}
where $\mathcal{E}_i$ are the Pauli errors with weight $i$. Usually ``weight enumerator" refers to a polynomial with these coefficients, but we will instead refer to the coefficient vectors $A = (A_0, \dots, A_n) $ and $B = (B_0, \dots, B_n) $ as the weight enumerators. 

Two codes are said to be \bfit{equivalent} if their code spaces are related by a non-entangling gate, i.e., a gate from $U(2)^{\otimes n} \rtimes S_n$, the local unitaries together with permutations. The weight enumerators form a code invariant. 
\begin{lemma}\cite{quantumMacWilliams,quantumweight}\label{lem:equivalentCodes}
Equivalent codes have the same weight enumerators $A$ and $B$.
\end{lemma} 

One of the most useful properties of weight enumerators is computing distance.
\begin{lemma}\cite{quantumMacWilliams} \label{lem:WtEnumDist}
The distance of an $((n,K))$ code is $d$ iff $A_i = B_i$ for all $i < d$.
\end{lemma}

\subsection{Exactly Transversal}

\begin{lemma}\label{lem:nodd}
    Exactly transversal $X$ and $Z$ implies $n$ is odd. 
\end{lemma}
\begin{proof}
The commutator of $ X^{\otimes n} $ and $ Z^{\otimes n} $ is $ (-1)^n $. The commutator of logical $ X $ and logical $ Z $ must be $ -1 $. Thus $ n $ must be odd.
\end{proof}

\subsection{Weights}
The Hamming weight of a bit string is the number of $ 1 $'s it has. For example, the bit string $ 000 $ has Hamming weight $ 0 $ while $ 011,101,110 $ all have have Hamming weight $ 2 $. So 
$
000,011,101,110 
$
are all of the even weight bit strings of length $ 3 $. And 
$
111,100,010,001 
$
are all the odd weight bit strings of length $ 3 $.
The weight of a computational basis ket is the Hamming weight of the bit string it corresponds to. So 
$$
\ket{000},\ket{011},\ket{101},\ket{110} 
$$
are all the even weight computational basis kets for $ 3 $ qubits. And 
$$
\ket{111},\ket{100},\ket{010},\ket{001} 
$$
are all the odd weight kets.
Now we turn our attention to Pauli operators.

Up to a global phase, every $ n $ qubit Pauli operator $ E $ is a tensor product of the form $E = g_1 \otimes \cdots \otimes g_n$ where each $ g_i $ is a 1-qubit Pauli gate $ g_i \in \{ I ,X, Y ,Z \} $. Define the \textit{weight} of $ E $, denoted $ wt(E) $, to be the number of $ g_i \in \{ X,Y,Z \}  $ in this tensor product.

Define the $ X $ weight, $ wt_X(E) $, to be the number of $ g_i \in \{ X,Y \} $ and define the $ Z $ weight, $ wt_Z(E) $, to be the number of $ g_i \in \{ Z,Y \} $. Define $ n_Y(E) $ to be the number of $ g_i=Y $. Similarly, define $ n_X(E) $ to be the number of $ g_i=X $ and define $ n_Z(E) $ to be the number of $ g_i=Z $.

The following lemma shows that $ n_Y(E) $ is a function of the other three weights.

\begin{lemma}\label{lem:wtY}
\[
    n_Y = wt_X + wt_Z - wt.
\]
\end{lemma}
\begin{pf}
Suppose the total weight of a Pauli string is $wt$. There are $ n_X=wt_X-n_Y $ many $g_i=X$ factors, $n_Z=wt_Z-n_Y $ many $ g_i=Z $, and $n_Y$ many $g_i=Y$. These must add up to the total weight $wt=n_X+n_Y+n_Z$. Thus $wt = (wt_X- n_Y) +  n_Y + (wt_Z - n_Y)$. Rearranging we have the claim.
\end{pf}

\subsection{Symmetric operators}

A matrix $M$ is said to be \bfit{symmetric} if $M = M^T$ and is said to be \bfit{anti-symmetric} if $M = - M^T$. The Pauli matrices $I$, $X$, and $Z$ are symmetric and $Y$ is anti-symmetric.

\begin{lemma}\label{lem:realsymmetricY}
Let $E = g_1 \otimes \cdots \otimes g_n$ be a Pauli string. Then $ E $ is symmetric iff $ n_Y(E) $ is even and $E$ is antisymmetric iff $ n_Y(E) $ is odd. 
\end{lemma}

\begin{lemma}\label{lem:antisym}
Let $E$ be an anti-symmetric Pauli string. That is, $ n_Y(E) $ is odd. Suppose $ \ket{\Psi} $ is real. Then
\[
\mel{\Psi}{E}{\Psi}=0.
\]

\end{lemma}
\begin{pf}
We have
\begin{align*}
    \mel{\Psi}{E}{\Psi} &\explainequals{1}\mel{\Psi}{E}{\Psi}^T \\
        &=\ket{\Psi}^T E^T \bra{\Psi}^T \\
    &\explainequals{2} \ket{\Psi}^\dagger E^T \bra{\Psi}^\dagger \\
    &= \mel{\Psi}{E^T}{\Psi} \\
        &\explainequals{3}-\mel{\Psi}{E}{\Psi}
\end{align*}
Line (1) uses the fact that every number is it's own transpose while line (2) uses the realness of $ \ket{\Psi} $. Line (3) uses that $E$ is antisymmetric. 

Thus we have $ \mel{\Psi}{E}{\Psi}=-\mel{\Psi}{E}{\Psi} $ and the claim follows.
\end{pf}

\section{Main Result}

In this section we will assume we have a real  code with $X$ and $Z$ exactly transversal. The code space has dimension $ K=2 $ so the code projector can be written as
$$
    \Pi=  \dyad{0} +  \dyad{1}.
$$

\begin{lemma}\label{lem:Xweight}

\begin{align*}
    wt_X(E) &=\text{odd} \implies \mel{0}{E}{0} =0= \mel{1}{E}{1} ,  \numberthis \label{eqn:Xodd} \\
    wt_X(E) &=\text{even} \implies \mel{0}{E}{1} =0= \mel{1}{E}{0} . \numberthis \label{eqn:Xeven} 
\end{align*}
\end{lemma}
\begin{pf}
  First notice that because $Z$ is exactly transversal then in the expansion of $\ket{0}$ with respect to the computational basis only even weight kets will have nonzero coefficients. Similarly, expanding $\ket{1}$ in the computational basis only odd weight kets will have nonzero coefficients. When an error $E$ has $wt_X(E)$ odd, then $E$ changes the parity of whatever basis ket it acts on, whereas if $wt_X(E)$ is even, then $E$ preserves parity. A linear combination of even weight kets is always orthogonal to a linear combination of odd weight kets. The claim follows.
\end{pf}

In a similar vein we have the following.  
\begin{lemma}\label{lem:Zweight}

\begin{align*}
    wt_Z(E) &=\text{even} \implies \mel{0}{E}{0} =\mel{1}{E}{1}  .
\end{align*}
\end{lemma}
\begin{proof} We can write
    \begin{align*}
        \mel{1}{E}{1}&=\mel{0}{XEX}{0}\\
        &\explainequals{1}\mel{0}{(-1)^{wt_Z(E)}EXX}{0}\\
        &=(-1)^{wt_Z(E)}\mel{0}{E}{0}\\
        &=\mel{0}{E}{0}.
    \end{align*}
    Lines $ (1) $ uses $XEX = (-1)^{wt_Z(E)} E$.
\end{proof}

\begin{lemma}\label{lem:Zweight}
\begin{align*}
    n_Y \text{ even and } wt_Z(E) \text{ odd} &\implies \mel{0}{E}{1} = 0, \\
n_Y \text{ odd and } wt_Z(E) \text{ even} &\implies \mel{0}{E}{1} = 0.
\end{align*}

\end{lemma}
\begin{pf}
Recall that $ n_Y(E) $ even is equivalent to $ E $ symmetric and $ n_Y(E) $ odd is equivalent to $ E $ antisymmetric. 
We have
\begin{align*}
    \mel{0}{E}{1} &=\mel{0}{EX}{0} \\
    &= (-1)^{wt_Z(E)}\mel{0}{XE}{0} \\
    &=(-1)^{wt_Z(E)}\mel{1}{E}{0} \\
        &\explainequals{1}(-1)^{wt_Z(E)}\mel{1}{E}{0}^T \\
        &=(-1)^{wt_Z(E)}\ket{0}^T E^T \bra{1}^T \\
    &\explainequals{2} (-1)^{wt_Z(E)}\ket{0}^\dagger E^T \bra{1}^\dagger \\
    &= (-1)^{wt_Z(E)}\mel{0}{E^T}{1}.
\end{align*}
In line (1), we use that each number is it's own transpose while line (2) uses the realness of the codewords.

If $E$ is symmetric then $E^T = E$. It follows in this case that if $wt_Z(E)$ is odd then $\mel{0}{E}{1} = 0$.

If $E$ is anti-symmetric then $E^T =-E$. In this case $\mel{0}{E}{1} = (-1)^{1+ wt_Z(E)} \mel{0}{E}{1}$. It follows that if $wt_Z(E)$ is even then $\mel{0}{E}{1} = 0$.
\end{pf} 
 
Using these lemmas, we can rewrite the $A_i$. 
\begin{lemma}\label{lem:Acoeffs}
\[
     A_i  = \sum_{ \substack{E \in \mathcal{E}_i: \\ wt_Z(E) \text{ even} \\ wt_X(E) \text{ even} }} |\mel{0}{E}{0}|^2.  
\]
\end{lemma}
\begin{pf}
Start with the $A$-type weight enumerators. Using the fact that $XEX = (-1)^{wt_Z(E)} E$ we have the following:
\begin{align*}
    \Tr{E \Pi} &=  \Tr{ E \qty( \dyad{0} +  \dyad{1} )}\\
    &=  \Tr{ E \qty( \dyad{0} + X \dyad{0} X)}\\
    &= \Tr{ E \dyad{0}} + \Tr{E X \dyad{0} X} \\
    &= \mel{0}{E}{0} + \Tr{X E X \dyad{0} }  \\
    &=   \mel{0}{E}{0} + (-1)^{wt_Z(E)} \Tr{E \dyad{0} }  \\
    &= \qty( 1 + (-1)^{wt_Z(E)})\mel{0}{E}{0}.
\end{align*}
If $wt_Z(E)$ is odd then this term cancels and if $wt_Z(E)$ is even then we get a factor of $2$. On the other hand, from \cref{lem:Xweight} we know that when $wt_X(E)$ is odd then $\mel{0}{E}{0} = 0$ and so we must also have that $wt_X(E)$ is even. It follows that 
\[
     A_i = \frac{1}{4} \sum_{E \in \mathcal{E}_i} |\Tr{E \Pi}|^2  = \sum_{ \substack{E \in \mathcal{E}_i: \\ wt_Z(E) \text{ even} \\ wt_X(E) \text{ even} }} |\mel{0}{E}{0}|^2.  
\]
\end{pf}

We now present our two main results.

\begin{theorem}\label{thm:Aodd} For a real code with $ X,Z $ 
 exactly transversal, $A_i = 0$ for all odd $i$. 
\end{theorem}
\begin{proof}
In \cref{lem:Acoeffs} we proved the following:
    \[
     A_i  = \sum_{ \substack{E \in \mathcal{E}_i: \\ wt_Z(E) \text{ even} \\ wt_X(E) \text{ even} }} |\mel{0}{E}{0}|^2.  
\]
 $wt_X(E)$ and $wt_Z(E)$ are even for all terms in the sum. Since $wt(E) = i$ is assumed odd then from \cref{lem:wtY} we have that $n_Y(E)$ is odd. This implies that $E$ is antisymmetric and so by \cref{lem:antisym} $ \mel{0}{E}{0} = 0 $. The claim follows.
\end{proof}

\begin{theorem}\label{thm:main}
For a real code with $ X,Z $ exactly transversal, $A_i = B_i$ for all even $i$.
\end{theorem}

\begin{proof} 

The $A$ enumerators are described in \cref{lem:Acoeffs}. Let's rewrite the $B$ coefficients as well:
\begin{align*}
    &\Tr{E\Pi E \Pi} \\ 
    =& \Tr{\; E\; \qty(\dyad{0} + X \dyad{0} X) \; E \; \qty(\dyad{0} + X \dyad{0} X) \;} \\
    =& \Tr{E \dyad{0} E \dyad{0} } + \Tr{ EX \dyad{0} X E \dyad{0}} \\ 
    &  \qquad + \Tr{ E \dyad{0} E X \dyad{0} X} \\
    & \qquad + \Tr{EX \dyad{0} X E X \dyad{0} X} \\
   =& \mel{0}{E}{0}\mel{0}{E}{0} + \mel{0}{E}{1}  \mel{1}{E}{0} \\
    & \qquad + \mel{1}{E}{0} \mel{0}{E}{1} +\mel{1}{E}{1}\mel{1}{E}{1} \\
     \explainequals{1}& \mel{0}{E}{0}\mel{0}{E}{0} + \mel{0}{E}{1}  \mel{1}{E}{0} \\
    & \qquad + \mel{1}{E}{0} \mel{0}{E}{1} +\mel{0}{E}{0}\mel{0}{E}{0} \\
     \explainequals{2}& 2|\mel{0}{E}{0}|^2 + 2|\mel{0}{E}{1}|^2.
\end{align*}
 Line (1) uses $XEX = (-1)^{wt_Z(E)} E$ twice. Line (2) uses hermiticity of $ E $, $ E=E^\dagger $.

% \newpage
   
Thus the $B$ weight enumerators are
\begin{align*}
  B_i &= \frac{1}{2} \sum_{E \in \mathcal{E}_i} \Tr{E \Pi E \Pi} \\
    &= \sum_{E \in \mathcal{E}_i} |\mel{0}{E}{0}|^2 + |\mel{0}{E}{1}|^2 \\
    &\explainequals{4}  \sum_{ \substack{E \in \mathcal{E}_i : \\ wt_X(E) = \text{even} }}  |\mel{0}{E}{0}|^2   +  \sum_{ \substack{E \in \mathcal{E}_i : \\ wt_X(E) = \text{odd} }}  |\mel{0}{E}{1}|^2   \\
 &=  \sum_{ \substack{E \in \mathcal{E}_i : \\ wt_X(E) = \text{even} \\ wt_Z(E) = \text{even} }}  |\mel{0}{E}{0}|^2   + \underbrace{  \sum_{ \substack{E \in \mathcal{E}_i : \\ wt_X(E) = \text{even} \\ wt_Z(E) = \text{ odd} }}  |\mel{0}{E}{0}|^2  }_{C_i} \\
 & \qquad \qquad  + \underbrace{  \sum_{ \substack{E \in \mathcal{E}_i : \\ wt_X(E) = \text{odd} }} |\mel{0}{E}{1}|^2  }_{D_i}. 
\end{align*}
In line (4) we use \cref{lem:Xweight}. To be sure, we have
\[
     B_i = A_i + \qty( C_i + D_i ).
\]
However, the next two lemmas show that $C_i = D_i = 0$ when $i$ is even.

\begin{lemma}
$C_i = 0$ for all even $i$.
\end{lemma}
\begin{pf}
For every term in the sum $wt_X(E)$ is even and $wt_Z(E)$ is odd. Since $wt(E) = i$ is even then from \cref{lem:wtY} we have that $n_Y(E)$ is odd. This implies that $E$ is anti-symmetric and so, by \cref{lem:antisym} $, \mel{0}{E}{0} = 0 $. The claim follows. 
\end{pf}

\begin{lemma}
$D_i = 0$ for all $i$ even.
\end{lemma}
\begin{pf}
    $wt_X(E)$ is odd for every term in the sum. Since $wt(E) = i$ is even then the difference $wt_X - wt$ is odd. Combining this with \cref{lem:wtY} we conclude that $n_Y$ and $wt_Z$ have opposite parity. So by \cref{lem:Zweight} we have that $\mel{0}{E}{1} =0$ for every term in the sum. The claim follows.
\end{pf}

This completes the proof of the theorem.
\end{proof} 

We immediately have the following consequence. 
\begin{corollary} A real code with $X$ and $Z$ exactly transversal must have odd distance $d$.
\end{corollary}
In other words, every error detecting code is actually error correcting.

% \ \newpage 
\section{Applying the main result to Stabilizer Codes}

\begin{lemma}\label{realStabilizerCode}
A stabilizer code is real if and only if $ n_Y $ is even for all the stabilizer generators. 
    
\end{lemma}

\begin{proof}
    Recall that a code is real if and only if the code projector is real, $ \Pi=\Pi^* $. For a stabilizer code, $ \Pi $ is a uniform sum over the stabilizer. So a code is real if and only if all Pauli gates in the stabilizer are real. Since $ I,X,Z $ are real while $ Y^*=-Y $ a Pauli stabilizer $ E $ is real if and only if $ wt_Y(E) $ is even. The result follows.
\end{proof}

\begin{lemma}\label{lem:XexactZexactStabilizerCode}
An $[[n,1,d]]$ stabilizer code has $ X $ and $ Z $ exactly transversal (or at least is equivalent to such a stabilizer code) if and only if $ n $ is odd and all stabilizer generators have $ wt_X $ even and $ wt_Z $ even. 
    
\end{lemma}

\begin{proof}
    Suppose $ X $ and $ Z $ are exactly transversal. Then by \cref{lem:nodd} we must have that $ n $ is odd. Furthermore, all stabilizer generators $ E $ commute with $ X^{\otimes n} $, so $ wt_Z(E) $ is even, and commute with $ Z^{\otimes n} $, so $ wt_X(E) $ is even.

    Now consider the reverse implication. Since all stabilizer generators are assumed to have $ wt_X $ even and $ wt_Z $ even then $ X^{\otimes n} $ and $ Z^{\otimes n} $ must be in the normalizer $ N(S) $. Since $ n $ is odd then $ X^{\otimes n} $ and $ Z^{\otimes n} $ anticommute and so $ X^{\otimes n} , Z^{\otimes n} \in N(S) \setminus S $  are not in the stabilizer and  they generate all the logical operations in $ N(S)/S $.  In this case, there will always exist a Clifford gate $ C $ such that $ C^{ \otimes n} $ gives an equivalent stabilizer code with $ X $ and $ Z $ exactly transversal. 
\end{proof}

For an example of the situation described at the end of \cref{lem:XexactZexactStabilizerCode}, recall that, in the most common version of the Shor code, $ X^{\otimes 9} $ implements logical $ Z $ and $ Z^{\otimes 9} $ implements logical $ X $. Then applying $ H^{ \otimes 9} $, where $ H= \frac{1}{\sqrt{2}} \begin{bmatrix} 1 & 1 \\ 1 & -1 \end{bmatrix} $ is the Hadamard gate, yields an equivalent stabilizer code with $ X $and $ Z $ exactly transversal.

\begin{lemma}\label{realAndExactStabilizerCode}
A stabilizer code is real and has $ X $ and $ Z $ exactly transversal (or at least is equivalent to such a stabilizer code) if and only if $ n $ is odd and all stabilizer generators have $ wt_X,wt_Z,n_Y,n_X,n_Z $ all even.
    
\end{lemma}

\begin{proof}
    The result follows from \cref{realStabilizerCode} and \cref{lem:XexactZexactStabilizerCode}.
\end{proof}

\cref{realAndExactStabilizerCode} provides an easy condition for when we can apply \cref{thm:main} to stabilizer codes. So we obtain the following proposition.

\begin{prop}
Suppose we have an $[[n,1,d]]$ stabilizer code with $n$ odd. Furthermore suppose that the number of $X$'s, $Y$'s, and $Z$'s in each stabilizer generator is even. Then $A_i = 0$ for odd $i$ and $A_i = B_i$ for even $i$. It follows that $d$ is odd. 
\end{prop}

\begin{pf}
    By \cref{realAndExactStabilizerCode} such a code must be (at least equivalent to) a real code with $ X $ and $ Z $ exactly transversal. Since equivalent codes have the same weight enumerator, by \cref{lem:equivalentCodes}, the result follows.
\end{pf}

\ \\
\section*{Acknowledgments} 
We would like to thank Michael Gullans and Victor Albert for helpful conversations. This research was supported in part by NSF QLCI grant OMA-2120757.

\appendix
\section{Weight Enumerators for various real codes with $X$ and $Z$ exactly transversal}

\begin{table}[htp]
    \centering
    \hspace*{-0.25cm}
    \begin{tabular}{|p{0.11\textwidth}|l|} \hline 
       $[[5,1,3]]$ & \begin{tabular}{l} $A = (1,0,0,0,15,0)$ \\ $B = (1,0,0,30,15,18)$ \end{tabular} \\  \hline 
        $[[7,1,3]]_\text{Steane}$ & \begin{tabular}{l} $A = (1,0,0,0,21,0,42,0)$ \\
    $B = (1,0,0,21,21,126,42,45)$
    \end{tabular} \\ \hline 
    $[[9,1,3]]_\text{Shor}$ & \begin{tabular}{l} 
    $A =  (1,0,9,0,27,0,75,0,144,0)$ \\
   $ B = (1,0,9,39,27,207,75,333,144,189)$
    \end{tabular} \\ \hline 
    $((11,2,3))$ \cite{nonadditiveexample} & \begin{tabular}{l} 
    $A = \qty(1,0,0,0,\frac{110}{3},0,88,0,605,0,\frac{880}{3},0)$ \\
    $B = \qty(1, 0, 0, \frac{55}{3}, \frac{110}{3}, 88, 88, 1210, 605, \frac{4400}{3}, \frac{880}{3}, 289 )$ 
    \end{tabular} \\  \hline 
    $[[11,1,5]]$ & \begin{tabular}{l} 
    $A = (1,0,0,0,0,0,198,0,495,0, 330 ,0 )$ \\
    $B = (1,0,0,0,0,198,198,990, 495,1650,330,234)$ 
    \end{tabular} \\ \hline 
    \end{tabular}
    \caption{}
\end{table}

Notice that each code in the table has $A_i = 0$ for odd $i$ and $A_i = B_i$ for even $i$. Also notice that $ n $ is odd and $d$ is odd. We have included the $((11,2,3))$ non-additive code to affirm that the results above hold for non-additive codes.

\bibliography{biblio}

\begin{thebibliography}{3}%
\makeatletter
\providecommand \@ifxundefined [1]{%
 \@ifx{#1\undefined}
}%
\providecommand \@ifnum [1]{%
 \ifnum #1\expandafter \@firstoftwo
 \else \expandafter \@secondoftwo
 \fi
}%
\providecommand \@ifx [1]{%
 \ifx #1\expandafter \@firstoftwo
 \else \expandafter \@secondoftwo
 \fi
}%
\providecommand \natexlab [1]{#1}%
\providecommand \enquote  [1]{``#1''}%
\providecommand \bibnamefont  [1]{#1}%
\providecommand \bibfnamefont [1]{#1}%
\providecommand \citenamefont [1]{#1}%
\providecommand \href@noop [0]{\@secondoftwo}%
\providecommand \href [0]{\begingroup \@sanitize@url \@href}%
\providecommand \@href[1]{\@@startlink{#1}\@@href}%
\providecommand \@@href[1]{\endgroup#1\@@endlink}%
\providecommand \@sanitize@url [0]{\catcode `\\12\catcode `\$12\catcode
  `\&12\catcode `\#12\catcode `\^12\catcode `\_12\catcode `\%12\relax}%
\providecommand \@@startlink[1]{}%
\providecommand \@@endlink[0]{}%
\providecommand \url  [0]{\begingroup\@sanitize@url \@url }%
\providecommand \@url [1]{\endgroup\@href {#1}{\urlprefix }}%
\providecommand \urlprefix  [0]{URL }%
\providecommand \Eprint [0]{\href }%
\providecommand \doibase [0]{https://doi.org/}%
\providecommand \selectlanguage [0]{\@gobble}%
\providecommand \bibinfo  [0]{\@secondoftwo}%
\providecommand \bibfield  [0]{\@secondoftwo}%
\providecommand \translation [1]{[#1]}%
\providecommand \BibitemOpen [0]{}%
\providecommand \bibitemStop [0]{}%
\providecommand \bibitemNoStop [0]{.\EOS\space}%
\providecommand \EOS [0]{\spacefactor3000\relax}%
\providecommand \BibitemShut  [1]{\csname bibitem#1\endcsname}%
\let\auto@bib@innerbib\@empty
%</preamble>
\bibitem [{\citenamefont {Shor}\ and\ \citenamefont
  {Laflamme}(1996)}]{quantumMacWilliams}%
  \BibitemOpen
  \bibfield  {author} {\bibinfo {author} {\bibfnamefont {P.}~\bibnamefont
  {Shor}}\ and\ \bibinfo {author} {\bibfnamefont {R.}~\bibnamefont
  {Laflamme}},\ }\href@noop {} {\bibinfo {title} {Quantum macwilliams
  identities}} (\bibinfo {year} {1996}),\ \Eprint
  {https://arxiv.org/abs/quant-ph/9610040} {arXiv:quant-ph/9610040 [quant-ph]}
  \BibitemShut {NoStop}%
\bibitem [{\citenamefont {Rains}(1996)}]{quantumweight}%
  \BibitemOpen
  \bibfield  {author} {\bibinfo {author} {\bibfnamefont {E.~M.}\ \bibnamefont
  {Rains}},\ }\href@noop {} {\bibinfo {title} {Quantum weight enumerators}}
  (\bibinfo {year} {1996}),\ \Eprint {https://arxiv.org/abs/quant-ph/9612015}
  {arXiv:quant-ph/9612015 [quant-ph]} \BibitemShut {NoStop}%
\bibitem [{\citenamefont {Roychowdhury}\ and\ \citenamefont
  {Vatan}(1997)}]{nonadditiveexample}%
  \BibitemOpen
  \bibfield  {author} {\bibinfo {author} {\bibfnamefont {V.~P.}\ \bibnamefont
  {Roychowdhury}}\ and\ \bibinfo {author} {\bibfnamefont {F.}~\bibnamefont
  {Vatan}},\ }\href@noop {} {\bibinfo {title} {On the structure of additive
  quantum codes and the existence of nonadditive codes}} (\bibinfo {year}
  {1997}),\ \Eprint {https://arxiv.org/abs/quant-ph/9710031}
  {arXiv:quant-ph/9710031 [quant-ph]} \BibitemShut {NoStop}%
\end{thebibliography}%

\end{document}